\documentclass[runningheads]{llncs}
\usepackage[T1]{fontenc}
\usepackage{graphicx}
\usepackage{latexsym}
\usepackage{amsmath}
\usepackage{mathrsfs}
\usepackage{amssymb}
\usepackage{epsfig}
\usepackage{color}
\usepackage{todonotes}
\usepackage[ruled,vlined,linesnumbered]{algorithm2e}
\usepackage{algpseudocode}
\usepackage{xspace}
\usepackage{float}
\usepackage{mathtools}
\usepackage{tabularx}
\usepackage{relsize}
\usepackage[pagewise]{lineno}
\usepackage{enumitem}
\usepackage{cite}

\numberwithin{mytheorem}{section}
\newtheorem{mylemma}{Lemma}[section]
\numberwithin{mylemma}{section}

\numberwithin{obs}{section}

\numberwithin{clm}{mylemma}

\let\oldnl\nl
\newcommand{\nonl}{\renewcommand{\nl}{\let\nl\oldnl}}

\newtheorem{innercustomgeneric}{\customgenericname}
\providecommand{\customgenericname}{}

\usepackage{todonotes}
\usepackage{complexity}

\newtheorem{preprocessing rule}{Preprocessing Rule}
\newtheorem{reduction rule}{Reduction Rule}
\newtheorem{branching rule}{Branching Rule}

\begin{document}
\title{3-Coloring $C_4$ or $C_3$-free Diameter Two Graphs\thanks{Tereza Klimo\v{s}ov\'{a} is supported by the Center for Foundations of Modern Computer Science (Charles Univ. project UNCE/SCI/004) and by GA\v{C}R grant 22-19073S.}}
%
%
\author{Tereza Klimošová\inst{1} \and
Vibha Sahlot\inst{2}
}
\authorrunning{T. Klimo\v{s}ov\'{a} and V. Sahlot}
%
\institute{Charles University in Prague, Czechia \and 
University of Cologne, Germany\\
\email{tereza@kam.mff.cuni.cz, sahlotvibha@gmail.com}
}
\maketitle              
\begin{abstract}
The question of whether 3-\textsc{Coloring} can be solved in polynomial-time for the diameter two graphs is a well-known open problem in the area of algorithmic graph theory. We study the problem restricted to graph classes that avoid cycles of given lengths as induced subgraphs.  Martin et. al. [CIAC 2021] showed that the problem is polynomial-time solvable for $C_5$-free or $C_6$-free graphs, and, $(C_4,C_s)$-free graphs where $s \in \{3,7,8,9\}$. We extend their result proving that it is polynomial-time solvable for $(C_4,C_s)$-free graphs, for any constant $s$, and for $(C_3,C_7)$-free graphs. Our results also hold for the more general problem List 3-Colouring. 

\end{abstract}

\section{Introduction}

In graph theory, \textsc{$k$-Coloring} is one of the most extensively studied problems in theoretical computer science. Here, given a graph $G(V,E)$, we ask if there is a function $c: V(G) \rightarrow \{1,2, \ldots k\}$  coloring all the vertices of the graph with $k$ colors such that adjacent vertices get different colors. If such a function exists, then we call graph $G$ {\em $k$-colorable}. The \textsc{$k$-Coloring} is one of Karp's 21 NP-complete problems and is NP-complete for $k \geq 3$ \cite{DBLP:Karp72}. 

The \textsc{$3$-Coloring} is NP-hard even on planar graphs \cite{properColoring}. It motivates to study \textsc{$3$-Coloring} under various graph constraints. For example, lots of research has been done on hereditary classes of graphs, i.e., classes that
are closed under vertex deletion \cite{chudnovsky2006strong, demaine2005algorithmic, golovach2017survey, kratochvil2011can, bonomo2018three}. It has also led to the development of many powerful algorithmic techniques. 

However, many natural classes of graphs are not hereditary, for example, graphs with bounded diameter. These graph classes are not hereditary as the deletion of a vertex may increase the diameter of the graph. The {\em diameter} of a given graph is the maximum distance between any two vertices in the graph. Graphs with bounded diameter are interesting as they come up in lots of real-life scenarios, for example, real-world graphs like Facebook. In this paper, we restrict our attention to \textsc{$3$-Coloring} on graphs with diameter two. We formally define problem $3$-\textsc{Coloring Diameter Two} as follows: Given an undirected diameter two graph  $G$, find if  there exists a  $3$-coloring of $G$.

The structure of the diameter two graphs is not simple, as adding a vertex to any graph $G$ such that it is adjacent to all other vertices, makes the diameter of the graph at most two. Hence, the fact that {$3$-Coloring} is NP-complete for general graph class implies that {$4$-Coloring} is NP-complete for diameter two graphs. 

Mertzios and Spirakis \cite{mertzios2013algorithms} gave a very non-trivial NP-hardness construction proving that  {$3$-Coloring} is NP-complete for the class of graphs with diameter three, even for triangle -free graphs. Furthermore, they presented a subexponential algorithm for $3$-\textsc{Coloring Diameter Two} for $n$-vertex graphs with runtime $2^{\mathcal{O}(\sqrt{n \log n)}}$. D\c ebski et. al. provided a further improved algorithm for $3$-\textsc{Coloring Diameter Two} on $n$-vertex graphs with runtime $2^{\mathcal{O}(n^{\frac{1}{3}} {\log}^2 n)}$.  $3$-\textsc{Coloring Diameter Two} has been posed as an open problem in several papers \cite{bodirsky2012complexity, broersma2013three, martin2019colouring, mertzios2013algorithms, paulusma2016open, debsk2022faster}.

The problem has been studied for various subclasses and  is known to be polynomial-time solvable for:
\begin{itemize}
\item  graphs that have at least one articulation neighborhood \cite{mertzios2013algorithms}.
\item  $(C_3, C_4)$-free graphs \cite{martin2019colouring}.
\item $C_5$-free or $C_6$-free graphs , $(C_4,C_s)$-free graphs where $s \in \{3,7,8,9\}$ \cite{martin2022colouring}.
\item $K^2_{1,r}$-free or $S_{1,2,2}$-free graphs, where $r \geq 1$ \cite{martin2019colouring}.
\end{itemize}

Continuing this line of research, we further investigate $3$-\textsc{Coloring} for $C_4$-free and $C_3$-free diameter 2 graphs. In particular, we consider the following two problems:
\begin{enumerate}
\item $3$-\textsc{Coloring ($C_4,C_{k}$)-Free Diameter Two}, where given an undirected  induced $(C_4, C_{s} )$-free diameter two graph $G$ for constant natural number $s$, we ask if there exists a  $3$-coloring of $G$.
\item $3$-\textsc{Coloring ($C_3,C_{7}$)-Free Diameter Two}, where given an undirected  induced $(C_3, C_{7} )$-free diameter two graph $G$, we ask if there exists a  $3$-coloring of $G$.
\end{enumerate}

In fact, we consider a slightly more general problem of \textsc{List 3-Coloring}. A {\em list assignment} on $G$ is a function $L$ which assigns to every vertex $u\in V(G)$ a list of admissible colours. A list assignment is a {\em list $k$-assignment} if each list is a subset of a given $k$-element set. The problem of \textsc{List 3-Coloring} is then to decide whether there is a coloring $c$ of $G$ that {\em respects} a given list $3$-assignment $L$, that is, for each vertex $u\in V(G)$, $c(u)\in L(u)$. This problem has also been considered in many of the previously mentioned works, in particular the aforementioned results from~\cite{martin2022colouring}, some of which we use as subroutines in our algorithms, hold for~\textsc{List 3-Coloring} as well.

The following two theorems summarize the main results of our paper.

\begin{theorem}\label{thm:1}
The $3$-\textsc{Coloring ($C_4,C_{s}$)-Free Diameter Two} is polynomial-time solvable for any constant $s$.
\end{theorem}

\begin{theorem}\label{thm:2}
The $3$-\textsc{Coloring ($C_3,C_{7}$)-Free Diameter Two} is polynomial-time solvable.
\end{theorem}

The paper is organized as follows. We define the terminology and notations used in this paper in Section \ref{sec:prelim}. 
We give some preprocessing rules 
 in Section \ref{Sec:preprocessingRule}. Next, we prove that  $3$-\textsc{Coloring ($C_4,C_{k}$)-Free Diameter Two} is polynomial-time solvable in Section \ref{sec:C_4C_k}. Afterward, we prove $3$-\textsc{Coloring $(C_3,C_{7})$-Free Diameter Two} is polynomial-time solvable in Section \ref{sec:C_3C_k}.

 \section{Preliminaries} \label{sec:prelim} 
 In this section, we state the graph theoretic terminology and notation used in this paper. The set of consecutive integers from $1$ to $n$ is denoted by $[n]$. The vertex set and the edge set of a graph $G$ are denoted by $V(G)$ and $E(G)$, respectively (or simply $V$ and $E$ when the underlying graph $G$ is clear from the context). By $|G|$, we denote the order of $G$, that is $\max\{|V(G)|,|E(G)|\}$. An edge between vertices $u$ and $v$ is denoted as $(u,v)$.  
 For an unweighted and undirected graph $G(V,E)$, we define {\em distance} $d(u,v)$ between two vertices $u,v \in V(G)$ to be the length of a shortest path between $u,v$, if $u$ is reachable from $v$, else it is defined as $+ \infty$. The length of a path is defined by the number of edges in the path.
 
 Let $f:A \rightarrow B$ be a function. Then, for any non-empty set $A' \subseteq A$, by $f(A')$, we denote the set $\{f(a)| a \in A'\}$. 

For a vertex $v \in V(G)$, its {\em neighborhood} $N(v)$  is the set of all vertices adjacent to it and its {\em closed neighborhood} $N[v]$ is the set $N(v) \cup \{v\}$. Moreover, for a set $A \subseteq V$, $N_A(v)=N(v) \cap A$, similarly, $N_A[v]=N_A(v) \cup \{v\}$  We define $N_G[S]=N(S)=\bigcup_{v \in S}N_G[v]$ and $N_G[S]=N[S]=N_G[S] \setminus S$ where $S \subseteq V(G)$. The {\em degree} of a vertex $v \in V(G)$, denoted by $deg_G(v)$ or simply $deg(v)$, is the size of $N_G(v)$. 
A complete graph on $q$ vertices is denoted by $K_q$. 

A graph $G'$ is a \emph{subgraph} of $G$ if $V(G')\subseteq V(G)$ and $E(G')\subseteq E(G)$. A graph $G'$ is an \emph{induced subgraph} of $G$ if for all $x,y \in V(G')$ such that $(x,y)\in E(G)$, then $(x,y)\in E(G')$. For further details on graphs, refer 
to~\cite{diestel2000graph}.

We say that list assignment $L$ is a {\em $k$-list assignment} if $|L(v)| \leq k$ for each vertex $v \in V(G)$. 

In {\sc $k$-List Coloring}, given a graph $G$ and a $k$-list assignment $L$, we ask if $G$ has a coloring that respects $L$.

\begin{theorem} \cite{EDWARDS1986337} \label{2ListColPoly}
The {\sc $2$-List Coloring} is linear-time solvable.
\end{theorem}

Next, we have the following proposition that we may use without explicitly referring to these in the rest of our paper.

\begin{proposition}\label{lemma:1Obs}
 In $3$-\textsc{Coloring Diameter Two}, if the given diameter two graph $G$ has a vertex that has its neighborhood colored with at most three colors, then the instance is polynomial-time solvable in $|V(G)|$.
\end{proposition}
\begin{proof}
Suppose $v \in V(G)$ such that its neighborhood is completely colored. If $N(v)$ is colored with at most two colors, then we can assign one of the remaining color to $v$, else there is no valid $3$-coloring with the given color assignment of the neighborhood of $v$. Now, as $G$ has diameter two, $N(N(v)) \cup N(v) = V(G)$. Hence, we have \textsc{$2$-List Coloring} instance which can be solved in polynomial-time by Theorem \ref{2ListColPoly}.
\qed
\end{proof} 

Similarly, we can assume that the given graph does not contain a vertex with constant degree (or degree such that bruteforcing the assignment of the colors on its neighborhood does not exceed the time complexity we are aiming for).

\section{Prepoccessing Rules for $3$-\textsc{Coloring Diameter Two}} \label{Sec:preprocessingRule}

A {\em preprocessing rule} is a rule which we apply to the given instance to produce another instance or an answer YES or NO. It is said to be {\em safe} if it applying it to the given instance produces an equivalent instance. We say that a preprocessing rule is {\em applicable} on an instance if the output is different from the input instance. Now we list the preprocessing rules that we will use in later sections. 

Consider a diameter two graph $G$ with a list $3$-assignment such that each vertex $v \in V(G)$ is assigned a list (a set) $L(v)$ of colors from the set $\{a,b,c\}$. When $|L(v)|=1$ for some vertex $v$, we say that $v$ is {\em colored} and let $c(v)$ be the only element of $L(v)$.

\begin{preprocessing rule}\label{rr:propagate}
If $|L(u)|=1$, for every neighbor $v$ of $u$, let $L(v):=L(v)\setminus L(u)$.
\end{preprocessing rule}

\begin{preprocessing rule}\label{rr:emptyList}
If $L(v)= \emptyset$ for any vertex $v \in V(G)$, then $G$ is not list $k$-colorable.
\end{preprocessing rule} 

\begin{preprocessing rule}\label{rr:List2Col}
If $0<|L(v)| \leq 2$ for all vertices $v \in V(G)$, then {\sc $2$-List Coloring} is linear-time solvable in $|G|$ (by Theorem \ref{2ListColPoly}).
\end{preprocessing rule}

We call $K_4$ minus an edge a \emph{diamond}.

\begin{preprocessing rule}\label{rr:diamondCol}
If $G$ contains a diamond  $\{v,w,x,y\}$ such that it does not contain the edge $(v,x)$, let $L(x)=L(v):=L(v) \cap L(x)$.
\end{preprocessing rule} 

\begin{preprocessing rule}\label{rr:triangleCol}
If $G$ contains a triangle  $\{v,w,x\}$ such that $|L(v)|=2$ and  $L(v)=L(w)$, let $L(x)=L(x) \setminus L(v)$.
\end{preprocessing rule} 

\begin{preprocessing rule}\label{rr:4Cycle}
If $G$ contains an induced $C_4$  $\{v,w,x,y\}$ such that $|L(v)|=|L(w)|=|L(x)|=2$ and $L(v), L(w), L(x)$ are pairwise different, $L(y):=L(y)\setminus (L(v)\cap L(x))$.
\end{preprocessing rule}

\begin{proposition}
The {\it Preprocessing Rules} \ref{rr:propagate}, \ref{rr:emptyList}, \ref{rr:List2Col}, \ref{rr:diamondCol}, \ref{rr:triangleCol} and \ref{rr:4Cycle} are safe.
\end{proposition}
\begin{proof}
The Preprocessing Rules \ref{rr:propagate}, \ref{rr:emptyList}, \ref{rr:List2Col}, and \ref{rr:triangleCol} are easy to see. 

Safety of Preprocessing Rule~\ref{rr:diamondCol} follows from the fact
that any 3-colouring assigns $v$ and $x$ the same colour. Similarly, Safety of Preprocessing Rule~\ref{rr:triangleCol} follows from the fact that any 3-colouring assigns $x$ a color not in $L(v)$, since both colors in $L(v)=L(w)$ are used to color $v$ and $w$.

Now we consider Preprocessing Rule \ref{rr:4Cycle}. 
Without loss of generality, assume $L(v)=\{b,c\}, L(w)=\{a,c\}, L(x)=\{a,b\}$. If $w$ is colored  $a$, then $x$ will be colored $b$. Else if $w$ is colored $c$, then $v$ is colored $b$. Thus, one of $v,x$ is always colored $b$. Hence, $y$ cannot be colored $b$. Thus proved.

\end{proof}

\section{Polynomial-time algorithm for $3$-\textsc{Coloring ($C_4,C_{k}$)-Free Diameter Two}} \label{sec:C_4C_k} 

In this section, we prove that $3$-\textsc{Coloring ($C_4,C_{s}$)-Free Diameter Two}, for any constant $s$, has a polynomial-time algorithm. We reinstate the theorem.

\paragraph{Theorem~\ref{thm:1}.}
{\em $3$-\textsc{Coloring ($C_4,C_{s}$)-Free Diameter Two} is polynomial-time solvable for any constant $s$.}
\newline

Consider graph $G$ with a list 3-assignment $L$, where $L(v)=\{a,b,c\}$ for all $v \in V(G)$ initially (notice that in case of {\sc List $3$-Coloring}, we can initialise with any given list 3-assignment $L$).  

We may assume that $G$ contains an induced $C_5$, otherwise, we can solve the problem in polynomial-time as $3$-\textsc{Coloring} on $C_5$-free diameter two graphs is polynomial-time solvable \cite{martin2022colouring}. 
Consider a $C_5$ as $C_5^1=(1,2,3,4,5,1)$ in $G$. Note that all colorings of $C_5$ are equivalent up to renaming and cyclic ordering of the colors. 
Without loss of generality, assume $c(1)=a$, $c(2)=b$, $c(3)=a$, $c(4)=b$ and $c(5)=c$.

Let the open neighborhood of vertices in $C_5^1$, that is, $N(C_5^1)=N_1$ and the remaining vertices except for the vertices in $C_5^1$ are $N(N_1) \setminus C_5^1=N_2$. 
As $G$ has diameter $2$, hence $V-(C_5^1 \cup N_1 \cup N_2)=\emptyset $. 
Let $Col_1$ be the set of  vertices in $N_1$ that has list size one (i.e. their color is directly determined by the coloring of $C_5^1$ or during the algorithm). Similarly, $Col_2$ are the vertices of $N_2$ that have list size one.

As vertices of $C_5^1$ are already colored, the size of the list for the vertices in $N_1$ is at most two. Consider $N_1(i)=N(i) \setminus (C_5^1 \cup Col_1)$ for all $i \in [5]$. For example, $N_1(1)$ is the open neighborhood of $1$ except for the neighbors in $C_5^1$ and $Col_1$. Let $L_2 \subseteq N_2$ be the set of vertices with list size two and $L_3 \subseteq N_2$ be the set of vertices with list size three.

Throughout the algorithm (or in lemmas below) we assume that the neighborhood of any vertex is not fully colored, else by Proposition \ref{lemma:1Obs}, we can solve the problem in polynomial-time.

In the lemmas below, we assume that $G$ is a $(C_4, C_s)$-free diameter two graph with a list $3$-assignment $L$ on which Preprocessing rules have been exhaustively applied.

\begin{mylemma} \label{lem:7MoreConnectedComponents} 
If there are at most $k \in \mathbb{N}$ connected components in some $N_1(i)$ for $i \in [5]$, then list $3$-coloring can be resolved by solving at most $2^k$ instances of \sc{$2$-List Coloring}.
\end{mylemma}

\begin{proof}
Without loss of generality, assume that $i=1$. Notice that for any valid $3$-coloring of graph $G$, each connected component in $N_1(1)$ should be a bipartite graph. For contradiction, assume that there is an odd cycle in a connected component of $N_1(1)$. It requires at least three colors to color any odd cycle. But all the vertices in the odd cycle are adjacent to $1$. Hence, we require a fourth color to color $1$. This is a contradiction for any valid list $3$-coloring of $G$.

Now for each connected component in $N_1(1)$, arbitrarily choose a vertex and consider both possibilities of colors in its list. Propagate the coloring to the rest of the vertices in that connected component. As there are at most $k$ connected components in $N_1(1)$ by the assumption, we have $2^k$ possibilities of coloring all the vertices in $N_1(1)$. By Proposition \ref{lemma:1Obs}, we are left with  \textsc{List 2-Coloring} instance for each of the possible color assignments to $N(1)$. Thus, list $3$-coloring can be resolved by solving at most $2^k$ instances of \sc{$2$-List Coloring}.
\qed
\end{proof} 

\begin{mylemma}\label{lemma:1N1ListSize2AdjOneC5Vertex}
Each vertex in $N_1$ that has a list of size two is adjacent exactly to one vertex in $C_5^1$.
\end{mylemma}
\begin{proof}
Consider a vertex $v$ in $N_1 \setminus Col_1$, thus $|L(v)|=2$. Then there are two possibilities. Either it is adjacent to more than one differently colored vertices in $C_5^1$  or it is adjacent to more than one same colored vertices in $C_5^1$, say $i$ and $i+2$ for $i \in \{1,2\}$.  The first case implies that $|L(v)|=1$, which is a contradiction. In the second case, as $G$ is $C_4$-free, $i$, $i+1$, $i+2$ and $v$ either form a $K_4$ (which implies $G$ is not list $3$-colorable) or a diamond, where $i+1$ is a common neighbor of $i$ and $i+2$ in $C_5^1$. But by the Preprocessing Rule \ref{rr:diamondCol}, the diamond will imply that $v \in Col_1$, which is a contradiction. Hence, the second case is not possible and $v$ is adjacent to exactly one vertex in $C_5^1$.
\qed
\end{proof} 

\begin{mylemma} \label{lem:2list2NonAdjacent} 
There are no edges between vertices in $N_1(i)$ and in $N_1(i+1)$ for all $i \in [4]$ and between $N_1(1)$ and $N_1(5)$.
\end{mylemma}

\begin{proof}
Consider a vertex $v \in N_1(1)$ and a vertex $u \in N_1(2)$. Let $(u,v) \in E(G)$. As $G$ is $C_4$-free, therefore, the cycle $(1,v,u,2,1)$ has an edge, $(1,u)$ or $(2,v)$. This will reduce the size of the list of $u$ or $v$, respectively, to one by the Preprocessing Rule \ref{rr:diamondCol}. But this is a contradiction to the fact that $u,v \notin Col_1$. Similar arguments work for the remaining cases. 
\qed
\end{proof} 

\begin{mylemma} \label{lem:3atmostOneNeighbor} 
Every vertex in $N_1(i)$ has at most one neighbor in $N_1(j)$ for all $i,j \in [5]$ and $i \neq j$.
\end{mylemma}

\begin{proof}
Suppose not. Assume a vertex $v \in N_1(1)$ is adjacent to two vertices $x,y \in N_1(j)$ where $j$ can only be $3,4$ from the Lemma \ref{lem:2list2NonAdjacent}. As $G$ is $C_4$ free, the cycle $(v,x,j,y,v)$ should have a chord. By Lemma \ref{lemma:1N1ListSize2AdjOneC5Vertex}, $(v,j) \notin E(G)$. Thus, $(x,y) \in E(G)$, which is a contradiction as it implies $c(v)=c(j)$ by the Preprocessing Rule \ref{rr:diamondCol}, but $v \notin Col_1$. Similar arguments hold for the remaining cases.
\qed
\end{proof}

\begin{mylemma} \label{lem:4perfectMatching} 
We have $|N_1(1)|=|N_1(3)|$ and $|N_1(2)|=|N_1(4)|$. Also, $G[N_1(1),$ $N_1(3)]$, $G[N_1(2),N_1(4)]$ are perfect matchings.
\end{mylemma}

\begin{proof}
Consider a vertex $v \in N_1(1)$ and a vertex $u \in N_1(3)$. As, $v$ is not adjacent to $2$ or $4$, using Lemma \ref{lemma:1N1ListSize2AdjOneC5Vertex}, $v$ should be adjacent to some neighbor of $3$ in $N_1$ to keep the distance between $v$ and $3$ as at most two. Also, by Lemma \ref{lem:3atmostOneNeighbor}, it can have at most one neighbor in $N_1(3)$. This holds for all vertices in $N_1(1)$. Hence the graph induced on $N_1(1)$ and $N_1(3)$, i.e., $G[N_1(1),N_1(3)]$ is a perfect matching and $|N_1(1)|=|N_1(3)|$. Another case can be proven similarly.
\qed
\end{proof}

\begin{mylemma} \label{lem:5N2atmostOneNeighL3oneNeigh}
Every vertex in $N_2 \setminus Col_2 $ has at most one neighbor in $N_1(i), \; \forall i \in [5]$ and every vertex in $L_3 $ has exactly one neighbor in $N_1(i), \; \forall i \in [5]$.
\end{mylemma} 

\begin{proof}
Consider a vertex $v \in N_2$ that is adjacent to two vertices $z_1,z_2 \in N_1(1)$. Then $(z_1,z_2) \in E(G)$ as otherwise $(1,z_1,v,z_2,1)$ forms a $C_4$ but $G$ is $C_4$-free. This implies that the color of $v$ is the same as the color of the vertex $1$ by the Preprocessing Rule \ref{rr:diamondCol}, which is a contradiction as $v \notin Col_2$. 
Analogous arguments can be extended for the remaining cases. Hence, any vertex in $N_2 \setminus Col_2$ can be adjacent to at most one neighbor in $N_1(i), \; \forall i \in [5]$.

Suppose $v \in L_3$ and $v$ is not adjacent to any vertex in $N_1(1)$. As the diameter of the graph is two, the distance between $v$ and $1$ is at most two. This implies there is a vertex $y \in Col_1$ such that $(v,y),(y,1) \in E(G)$ by \ref{lemma:1N1ListSize2AdjOneC5Vertex}. But this reduces the list size of $L(v)$ to at most two as now $v$ is adjacent to a colored vertex. This is a contradiction as $|L(v)|$ is three. Hence, $v$ is adjacent to a vertex in $N_1(1)$. We can argue similarly for the remaining cases. Thus every vertex in $L_3 $ has exactly one neighbor in $N_1(i), \; \forall i \in [5]$.
\qed
\end{proof}

\begin{mylemma} \label{lem:6N1_1adjN_3NoCommonNeighbrL2L3}
Any pair of vertices $x \in N_1(1)$ and $y \in N_1(3)$ such that $(x,y) \in E(G)$, don’t share a common neighbour in $L_2$ or $L_3$. Similarly, any pair of vertices $w \in N_1(2)$ and $z \in N_1(4)$ such that $(w,z) \in E(G)$, don’t share a common neighbour in $L_2$ or $L_3$.
\end{mylemma}

\begin{proof}
Suppose not and there is a vertex in $ v \in L_2\cup L_3$ that is adjacent to both $x$ and $y$, for any two vertices $x \in N_1(1)$ and $y \in N_1(3)$ such that $(x,y)\in E(G)$. Both $x,y$ have list $\{b,c\}$. Hence $v$ is colored $a$,
 which is a contradiction as $ v \in L_2\cup L_3$. Another case can be proved using similar arguments.
\qed
\end{proof}

\begin{mylemma} \label{lem:8}
Let $z \in L_3$ and $u \in N_1(i)$ for  some $i \in [5]$ such that $uz \notin E(G)$. Then there is at most one vertex $z' \in N_2 \setminus \{z\} $ such that $(z,z'),(u,z') \in E(G)$. 
\end{mylemma}

\begin{proof}
For contradiction, assume that there are at least two common neighbors $z',z'' \in N_2 \setminus \{z\}$ of $u$ and $z$. Then, $(u,z',z,z'',u)$ forms a diamond (there exists the edge $(z',z'')$) as $G$ is $C_4$-free and $uz \notin E(G)$ by assumption). This implies that the size of the list of $z$ is two as the size of list of $u $ is two by the Preprocessing Rule \ref{rr:diamondCol}. This is a contradiction to the assumption that $z \in L_3$.
Hence proved.
\qed
\end{proof}

\begin{mylemma} \label{lem:9}
Either $G[L_2 \cup L_3]$ contains an induced path $P_{\ell}* $ of length $\ell-1$ for some $\ell \in \mathbb{N}$, or whether $G$ is list $3$-colorable can be decided by solving at most $\mathcal{O}(3^{6 \ell})$ {\sc $2$-List Coloring} instances. Here $P_{\ell}* = (p_1,p_2, \ldots p_{\ell})$ is such that the neighborhood of $p_1$ and $p_{\ell}$ in $N_1$ is disjoint from neighborhood of vertices $p_2,p_3 \ldots p_{\ell -1}$. 
\end{mylemma}

\begin{proof} 

Pick any vertex $p_1 \in L_3$. Let $j=0$. Repeat the following, until $P_{\ell}* $ is constructed or step~\ref{Step:2} fails. 
In the later case, we claim that whether $G$ is list $3$-colorable can be decided by solving at most $\mathcal{O}(3^{6 \ell})$ {\sc $2$-List Coloring} instances. Note that during the following, we only modify the lists, not the sets $L_3$, $N_1$, etc.

For $i=2j+1$:
\begin{enumerate}
\item \label{Step:1} Color $p_i$ and its five neighbors in $N_1$ and apply Preprocessing rules exhaustively.  
\item \label{Step:2} 
If there are vertices $x\in L_3$ and $y\in N_2$ satisfying the following:
\begin{enumerate}[label=(\roman*)]
\item \label{con:1} $x$ has a list of size three,
\item \label{con:2} $y$ is a common neighbor of $x$ and $p_i$ and $|L(y)|=2$, and 
\item \label{con:3} neighbors of $y$ in $N_1$ are not adjacent to $p_1$, 
\end{enumerate}
set $p_{i+1}=y$, $p_{i+2}=x$, color  $p_{i+1}$ and its (at most five) neighbors in $N_1$, increase $j$ by one, apply Preprocessing rules exhaustively and proceed to the next iteration.
\end{enumerate}

Note that if step~\ref{Step:2} fails because there is no $x$ satisfying~\ref{con:1}, we have a {\sc $2$-List Coloring} instance. If there is such $x$, since $G$ has diameter two, $p_i$ and $x$ have a common neighbor $y$. We next argue that every such $y$ satisfies~\ref{con:2}. As $x$ has a list of size three, it has no colored neighbors and since neighbors of $p_i$ in $N_1$ are colored, it follows that $y\in N_2$ and moreover, $L(y)\geq 2$. On the other hand, $L(y)\leq 2$ as it does not contain $c(p_i)$.

Before discussing the case when step~\ref{Step:2} fails because there is no pair of $x$ and $y$ satisfying~\ref{con:3}, we make a few observations about adjacencies in $G$.

First, observe that from the fact that $p_{2j+1}$, $j\geq 1$, was chosen as a vertex with list of size three, it follows that it is adjacent to none of the already colored vertices, in particular, to none of $p_1,\ldots, p_{2j}$ and their neighbors in $N_1$.

\paragraph{ Claim~1.}
{\em In the above procedure for $p_i$, where $i \neq 1$ and $i$ is odd, every neighbor of $p_1$ is adjacent to at most one neighbor of $p_i$.}

\noindent\textit{Proof.} Suppose there are at least two such common neighbors $s$ and $t$ of $p_1$ and $p_i$. Hence, $(p_1,s,p_i,t,p_1)$ forms a diamond with the edge $(s,t)$ (since  $G$ has no induced $C_4$ and $(p_1,p_i)$ is not an edge). By the Preprocessing Rule \ref{rr:diamondCol}, $L(p_i):=L(p_1)$ which is a contradiction with the fact that $p_i$ has a list of size three after coloring $p_1$ and applying Preprocessing rules.
\qed


\paragraph{ Claim~2.}
{\em In the above procedure for $p_i$, where $i \neq 1$ and $i$ is odd, $p_i$ has at most five neighbors in $N_2$ adjacent to neighbors of $p_1$ in $N_1$.}

\noindent\textit{Proof.} Assume $q \in N_1$ is a neighbor of $p_1$ adjacent to two neighbors $v,w \in N_2$ of $p_i$. Hence, $(q,v,p_i,w,q)$ forms a diamond with the edge $(v,w)$. Again, application of Preprocessing Rule \ref{rr:diamondCol} after coloring $p_1$ implies $|L(p_i)|=|L(p_1)|=1$, contradicting the choice of $p_i$.

\qed 



So if in any iteration step~\ref{Step:2} fails because of~\ref{con:3}, from Claim~2 it follows that all vertices with list of size three are adjacent to one of at most five neighbors of $p_i$ in $N_2$ which are adjacent to neighbors of $p_1$ in $N_1$. Thus, any coloring of these at most five vertices yields an instance of {\sc $2$-List Coloring}.

In total, if the process stops before constructing $P_{\ell}* $, at most $6\ell$ vertices are colored before reaching a {\sc $2$-List Coloring} instance (including vertices colored if step~\ref{Step:2} fails because of~\ref{con:3}). For each such vertex, we have at most three possible choices of color. So, we can decide whether the instance is list $3$-colorable by solving at most $\mathcal{O}(3^{6 \ell})$ {\sc $2$-List Coloring} instances or we construct $P_{\ell}* $.
\qed
\end{proof}

\noindent\textit{Proof for Theorem \ref{thm:1}. } As, $G$ is $(C_4,C_s)$-free, then, we can't have a $P_{\ell}*=(p_1,p_2, \ldots p_{\ell})$ where $\ell =s-4$ in Lemma \ref{lem:9} (i.e. the neighborhood of $p_1$ and $p_{\ell}$ in $N_1$ is disjoint from neighborhood of vertices $p_2,p_3 \ldots p_{\ell -1}$), otherwise, we can construct a $C_s=(p_1,p_2, \ldots p_{\ell},c_3,3,2,c_2,p_1)$, where $c_2 \in N_{N_1(2)}(p_1)$ and $c_3 \in N_{N_1(3)}(p_{\ell})$. Thus, we can decide whether the instance is $3$-colorable by solving at most $\mathcal{O}(3^{6 s})$ {\sc $2$-List Coloring} instances. For a constant $s$, the runtime is polynomial using Theorem \ref{2ListColPoly}.
\qed

\section{Polynomial-time algorithm for  $3$-\textsc{Coloring ($C_3,C_{7}$)-Free Diameter Two} } \label{sec:C_3C_k}

In this section, we prove that $3$-\textsc{Coloring ($C_3,C_{7}$)-Free Diameter Two} has a polynomial-time algorithm. We reinstate the theorem.

\paragraph{Theorem~\ref{thm:2}.}
{\em $3$-\textsc{Coloring ($C_3,C_7$)-Free Diameter Two} is polynomial-time solvable.}
\newline

Consider graph $G$ with a list 3-assignment $L$, where $L(v)=\{a,b,c\}$ for all $v \in V(G)$ initially (notice that in case of {\sc List $3$-Coloring}, we can initialise with any given list 3-assignment $L$). Similar to the previous section, in our algorithm, we try to reduce the size of list of vertices to get {\sc $2$-List Coloring}  instance.

$G$ has a $C_5$, otherwise, we can solve the problem in polynomial-time as $3$-\textsc{Coloring} on $C_5$-free diameter two graphs is polynomial-time solvable \cite{martin2022colouring}. 
Consider a $C_5$ as $C_5^1=(1,2,3,4,5,1)$ in $G$ and assume $c(1)=a$, $c(2)=b$, $c(3)=a$, $c(4)=b$ and $c(5)=c$.

Let the open neighborhood of vertices in $C_5^1$, that is, $N(C_5^1)=N_1$ and the remaining vertices except for the vertices in $C_5^1$ are $N(N_1) \setminus C_5^1=N_2$. 
Let $Col_1$ and $Col_2$ be the set of  vertices in $N_1$ and $N_2$, respectively,  that have list size one.

As the vertices of $C_5^1$ are already colored, the size of the list for the vertices in $N_1$ is at most two. 
Consider $A=(N(1) \cup N(3)) \setminus (C_5^1 \cup Col_1)$, $B=(N(2) \cup N(4)) \setminus (C_5^1 \cup Col_1)$ and $C=N(5)  \setminus (C_5^1 \cup Col_1)$. 
We further partition $A$ into $A_1,A_3$ and $A_{13}$, where the vertices in $A_1$ are adjacent to $1$ but not $3$, the vertices in $A_3$ are adjacent to $3$ but not $1$ and the vertices in $A_{13}$ are adjacent to both $1$ and $3$. 
Similarly, we partition $B$ into $B_2,B_4$ and $B_{24}$, where the vertices in $B_2$ are adjacent to $2$ but not $4$, the vertices in $B_4$ are adjacent to $4$ but not $2$ and the vertices in $B_{24}$ are adjacent to both $2$ and $4$.
We partition $N_2 \setminus Col_2$ into $L_3$ and $L_2$. The set $L_3$ contains the vertices that have list size three and $L_2$ contains the vertices that have list size two.

Throughout the algorithm (or in lemmas below) we assume that the neighborhood of any vertex is not fully colored, else by Proposition \ref{lemma:1Obs}, we can solve the problem in polynomial-time.

\begin{mylemma} \label{lem:C_3ABCList2} 
\begin{enumerate}
\item The vertices in $A$ are not adjacent to $2,4,5$. Similarly, vertices in $B$ are not adjacent to $1,3,5$ and vertices in $C$ are not adjacent to $1,2,3,4$.
\item Each of the vertex $v \in L_3$ has at least one neighbor in each $A,B$ and $C$.
\item The sets $A_1$, $A_3$, $A_{13}$, $B_2$, $B_4$, $B_{24}$ and $C$ are independent. Also, there is no edge between the vertices in $A_1$ and $A_{13}$, $A_3$ and $A_{13}$, $B_2$ and $B_{24}$, $B_2$ and $B_{24}$. 
\end{enumerate}
\end{mylemma}
\begin{proof}
Consider the first part of the lemma. As the vertices in $A$ have list of size two, thus they can't be adjacent to $2,4,5$. Similar arguments can be extended for the remaining cases.

Consider the second part of the lemma. Let $v \in L_3$. As the diameter of $G$ is two, there should be a common neighbor of $v$ and $1$ (or $3$) in $N_1$. But as $|L(v)|=3$, it can't be adjacent to any vertex in $Col_1$. Thus $v$ has at least one neighbor in $A$. More precisely, $v$ has at least one neighbor in each $A_1$ and $A_3$, or $v$ has at least one neighbor in $A_{13}$. Similar arguments can be extended for the remaining cases.

Consider the third part of the lemma. As $G$ is $C_3$-free, there cannot be an edge between neighbors of any vertex. Hence, $A_1$, $A_3$, $A_{13}$, $B_2$, $B_4$, $B_{24}$ and $C$ are independent sets. Similarly, the vertices in $A_1$ and $A_{13}$ are adjacent to $1$. Thus, there is no edge between the vertices in $A_1$ and $A_{13}$. Similar arguments can be extended for the rest of the cases.
\qed
\end{proof}

\begin{mylemma} \label{lem:C_3C_71stCase} 
The vertices in $A_3$ and $B_2$ don't have neighbors in $N_2 \setminus Col_2$ that sees $C$. Similarly, any vertex in $N_2 \setminus Col_2$ doesn't have neighbors in both $A_1$ and $B_4$.
\end{mylemma}
\begin{proof}
Suppose there exist vertices $z \in N_2 \setminus Col_2$, $u_b \in B_2$ and $u_c \in C$ such that $z$ is adjacent to both $u_b$ and $u_c$, then there is an induced $C_7$ $(u_b,z,u_c, 5, 4,3,2,u_b)$. To see this, notice that as $u_b$ and $u_c$ are both neighbors of $z$, hence $(u_b,u_c) \notin E(G)$. As per  Lemma \ref{lem:C_3ABCList2}, $u_b$ is not adjacent to $3,4,5$. Similarly, $u_c$ is not adjacent to $2,3,4$. As $z \in L_3$, it is not adjacent to $2,3,4,5$ by construction. This is a contradiction as $G$ is $C_7$-free.

Similarly, if there are vertices $u_a \in A_3$, $u_c' \in C$ and $z' \in N_2 \setminus Col_2$ such that $(z',u_a),(z',u_c') \in E(G)$, then there is an induced $C_7$ $(u_a,z',u_c', 5, 1,2,3,u_a)$. This can be verified using similar arguments as in the previous case. It is a contradiction as $G$ is $C_7$-free.

Likewise, if there are vertices $u_a \in A_1$, $u_b \in B_4$ and $z \in N_2 \setminus Col_2$ such that $z$ is is adjacent to both $u_a$ and $u_b$, then there is an induced $C_7$ $(u_a,z,u_b, 4, 3,2,1,u_a)$ based on similar arguments as in the previous cases. This is a contradiction as $G$ is $C_7$-free. Hence, the claim.
\qed
\end{proof}

\begin{mylemma} \label{lem:C_3z3} 
Any vertex $z \in N_2 \setminus Col_2$ that has a neighbor in $C$, neither sees any vertex in $A_3$, nor in $B_2$. Hence, any vertex $v \in L_3$ has neighbors both in $A_{13}$ and $B_{24}$.
\end{mylemma}
\begin{proof}
Suppose that $z$ has a neighbor $u_c \in C$. Assume that $u_a \in N_{A_3}(z)$. Then we have an induced  $C_7$ $(z,u_c,5,1,2,3,u_a,z)$. To see this, notice that $(u_a,u_c) \notin E(G)$ as both $u_a$ and $u_c$ are neighbors of $z$ and $G$ is $C_3$-free. As per  Lemma \ref{lem:C_3ABCList2}, $u_a$ is not adjacent to $1,2,5$ and $u_c$ is not adjacent to $1,2,3$. This is a contradiction as $G$ is $C_7$-free. Thus, $z$ does not have any neighor in $A_3$.

Now assume that $u_b \in N_{B_2}(z)$. Then we have an induced $C_7$ $(z,u_b,2,3,4,5,u_c,z)$ based on similar arguments. But it a contradiction as $G$ is $C_7$-free. Thus, $z$ does not have any neighor in $B_2$.

By Lemma \ref{lem:C_3ABCList2}, any vertex $v \in L_3$ has a neighbor in $C$. Hence, neighborhood of $v$ in $A_3$ and $B_2$ is empty. As $G$ has diameter two and $v$ has list size three, $v$ should have neighbor both in $A_{13}$ and $B_{24}$. 
\qed
\end{proof}

\begin{mylemma} \label{lem:C_3z3isolated} 
Vertices in $L_3$ are isolated in $G[N_2]$. 
\end{mylemma}
\begin{proof}
Consider a vertex $z \in L_3$. For contradiction, assume it has a neighbor $z' \in N_2$. Note that $z'\notin Col_2$. 

We first argue that $z'$ has no neighbors in $A_{13}$, $B_{24}$, and $C$. Assume that $z'$ has a neighbor $u_c \in C$. By Lemma \ref{lem:C_3z3}, $z$ has neighbors both in $A_{13}$ and $B_{24}$. Let $u_a \in N_{A_{13}}(z)$ and $u_b \in N_{B_{24}}(z)$. Then, $(u_a,u_c) \in E(G)$, else we have an induced $C_7$ $(u_a,3,4,5,u_c,z',z,u_a)$. Similarly, $(u_b,u_c) \in E(G)$, otherwise we have an induced $C_7$ $(u_b,2,1,5,u_c,z',z,u_b)$.  Consider the 4-cycle $(z,u_a,u_c,u_b,z)$. As, $u_a$,$u_b$ and $u_c$ have all different lists with list size two,by the Preprocessing Rule \ref{rr:4Cycle}, $z$ should have a list of size at most $2$, which contradicts that $z\in L_3$. The cases when $z'$ has a neighbor in $A_{13}$ or $B_{24}$ are analogous.


By Lemma \ref{lem:C_3C_71stCase} $z'$ does not have neighbors both in $A_1$ and $B_4$. Assume $z'$ does not have any neighbor in $B_4$, the other case is analogous.
Since $G$ has diameter two and $z'$ has no neighbor in $B_{24}$, $B_4$ and $C$, it has a neighbor in $Col_1$ adjacent to $5$ and a neighbor in $Col_1$ adjacent to $4$. Observe that since $|L(z')|\geq 2$ and $L(z')$ does not contain colors of colored neighbors of $z'$
, all neighbors of $z'$ in $Col_1$ have the same color, namely the color $a$, as it must be different from $c(4)=b$ and $c(5)=c$. It follows that $z'$ has no neighbor in $Col_1$ adjacent $1$ or $3$, since $c(1)=c(3)=a$.

Thus, since $z'$ has no neighbor in $A_{13}$, it has neighbors in both $A_1$ and $A_3$ (as diameter of $G$ is two). Let $a_1 \in N_{A_1}(z')$ and $a_3 \in N_{A_3}(z')$. Then $G$ contains an induced $C_7$ $(z', a_1,1,5,4,3,a_3,z')$. This is a contradiction as $G$ is $C_7$-free. Thus, any vertex $z \in L_3$ has no neighbor in $N_2$.
\qed
\end{proof}

\noindent\textit{Proof of Theorem \ref{thm:2}.} We argue that coloring any vertex $z_1\in L_3$ by color $c$, applying Preprocessing rules, then coloring any vertex $z_2$, which still has a list of size three (if it exists --- otherwise, we have a {\sc 2-List coloring} instance) by color $b$ and applying Preprocessing rules again, yields a {\sc 2-List coloring} instance. This leads to the following algorithm which requires resolving $O(|V^2|)$ instances of {\sc 2-List coloring} on $G$:
\begin{itemize}
\item resolve the {\sc 2-List coloring} instance  obtained by setting $L(z)=\{a,b\}$ for all $z\in L_3$, if it is a YES-instance, return YES
\item else: for all $z_1\in L_3$:
\begin{itemize}
	\item color $z_1$ by color $c$ and apply the Preprocessing rules exhaustively
	\item if the resulting instance is a {\sc 2-List coloring} instance, resolve it and if it is a YES-instance, return YES
	\item else: 
	\begin{itemize}
		\item resolve the {\sc 2-List coloring} instance obtained by setting $L(z')=\{a,c\}$ for all $z'$ with lists of size three, if it is a YES-instance, return YES
		\item else: for all $z_2$ with list of size three 
		\begin{itemize}
			\item color $z_2$ by $b$ and apply the Preprocessing rules exhaustively
			\item resolve the resulting {\sc 2-List coloring} instance, if it is a YES-instance, return YES
		\end{itemize}
	\end{itemize}
\end{itemize}
\item return NO
\end{itemize}

Note that in the following, the sets of vertices $L_3$, $Col_2$, $A$, $B$, $B_{24}$, etc., are not modified, coloring and application of Preprocessing rules change only the lists of colors available for the vertices.

Consider a vertex $z_1 \in L_3$ and color it $c$, (if no such vertex exists, then $G$ is a {\sc $2$-List Coloring} instance). 

Apply Preprocessing rules and assume that it does not yield a {\sc $2$-List Coloring} instance. Observe that all neighbors of $z_1$ in $A$ and $B$ are colored.  

Consider a vertex $z_2$ with list of size three
. It has a common neighbor $u_c$ with $z_1$ in $C$, as $z_2$ is not adjacent to any other neighbor of $z_1$  and $G$ has diameter two. 
Color $z_2$ by color $b$. Applying Preprocessing rules colors all neighbors of $z_2$ in $C$ and $B$. in particular, $u_c$ is colored $a$.
 
We claim that there is no vertex with list of size three in the resulting instance. For contradiction, assume there is such a vertex $z_3$. It has a common neighbor $u_b$ with $z_2$ in $B$, as $z_3$ is not adjacent to any other neighbor of $z_2$ and $G$ has diameter two. 

By Lemma \ref{lem:C_3z3} $z_3$ has a neighbor $u_a \in A_{13}$. Moreover, $z_3$ has a common neighbor with $z_1$ in $C \setminus \{u_c\}$, say $v_c$, as $z_3$ is not adjacent to any other neighbor of $z_1$, and $G$ has diameter two. 

By Lemma~\ref{lem:C_3z3}, $z_1$ has a neighbor $v_b$ in $B_{24}$. Notice that $(u_a,v_b) \in E(G)$, otherwise, we have an induced $C_7$ $(z_3,u_a,3,4,v_b,z_1,v_c,z_3)$ which is a contradiction as $G$ is $C_7$-free. But now, we have an induced $C_7$ $(z_1,u_c,z_2,u_b,$ $z_3,u_a,v_b,z_1)$ which is a contradiction as $G$ is $C_7$-free. Hence, we do not have such $z_3$. Thus, we must have reduced our given initial instance to some  {\sc $2$-List Coloring} instance (or a polynomial number of instances). Hence, $3$-\textsc{Coloring ($C_3,C_7$)-Free Diameter Two} is polynomial-time solvable.
\qed

\section{Conclusions}

We have proved that $3$-{\sc coloring} on diameter two graphs is polytime solvable for  $(C_4,C_s)$-free  graphs where $s$ is a constant, and  $(C_3,C_7)$-free graphs. In the first case, we give an FPT  on parameter $s$. Further, our algorithms also work for \textsc{List 3-Coloring} on the same graph classes.
This opens avenues for further research on this problem for general $C_4$-free or $C_3$-free graphs. A less ambitious question is to extend similar FPT results to  $(C_3,C_s)$-free  with parameter $s$.

\subsection*{Acknowledgements}
The authors would like to thank Kamak 2022 organised by Charles University for providing a platform for several fruitful discussions. 

\bibliographystyle{splncs04}

\bibliography{citations}

\end{document}